\documentclass[fleqn,10pt,twocolumn]{SICE_FES26}

\usepackage{amsmath,amssymb}
\usepackage{graphicx}
\usepackage{booktabs}
\usepackage{enumitem}
\usepackage{kotex}
\usepackage{array}
\usepackage{algorithm}
\usepackage{algpseudocode}
\usepackage{multirow}
\usepackage{color}
\usepackage{booktabs}
\usepackage{url}
\usepackage{siunitx}
\usepackage{tikz}
\usepackage{pgfplots}
\usepackage{grffile}
\pgfplotsset{compat=newest}
\usetikzlibrary{plotmarks}
\usetikzlibrary{arrows.meta}
\usepgfplotslibrary{patchplots}
\usepgfplotslibrary{groupplots}
\usetikzlibrary{shapes,arrows}
\tikzstyle{block} = [draw, fill=white, rectangle, minimum height=2em, minimum width=3em]
\tikzstyle{sum} = [draw, fill=white, circle, node distance=1cm]
\usetikzlibrary{decorations.pathmorphing}
\usetikzlibrary{calc}
\usetikzlibrary{arrows.meta,positioning}
\usetikzlibrary{calc}
\tikzset{
  block/.style = {draw, fill=white, rectangle, minimum height=2em, minimum width=4em, align=center},
  line/.style  = {-{Stealth[length=2.2mm]}, thick},
  plain/.style = {thick}
}
\usetikzlibrary{fit,backgrounds}

\newtheorem{assumption}{\bf Assumption}
\newtheorem{proposition}{\bf Proposition}

\newtheorem{lemma}{\bf Lemma}

\newcommand{\bbR}{\ensuremath{{\mathbb R}}}
\newcommand{\bbZ}{\ensuremath{{\mathbb Z}}}
\newcommand{\bbN}{\ensuremath{{\mathbb N}}}

\newcommand{\calC}{{\mathcal{C}}}

\newcommand{\calO}{{\mathcal{O}}}

\newcommand{\calK}{{\mathcal{K}}}

\newcommand{\sfr}{{\mathsf{r}}}
\newcommand{\sfs}{{\mathsf{s}}}

\newcommand{\bfx}{{\mathbf{x}}}
\newcommand{\bfK}{{\mathbf{K}}}

\newcommand{\bfu}{{\mathbf{u}}}

\newcommand{\Enc}{{\mathsf{Enc}}}
\newcommand{\Dec}{{\mathsf{Dec}}}
\newcommand{\Comm}{{\mathsf{Comm}}}
\newcommand{\sk}{{\mathsf{sk}}}
\newcommand{\ini}{{\mathsf{ini}}}
\newcommand{\modp}{~{\mathrm{mod}}~}

\title{Co-Design of Cryptographic Parameters and Delay-Aware Feedback Gain for Encrypted Control Systems}

\author{Yeongjun Jang${}^{1\dagger}$}
\speaker{Yeongjun Jang}

\affils{${}^{1}$ASRI, Department of Electrical and Computer Engineering, Seoul National University, Seoul, Korea\\
(Tel: +82-2-880-9053; E-mail: jangyj0512@snu.ac.kr)\\
}

\abstract{%
Encrypted control employs homomorphic encryption (HE) to protect both the computation and communication stages, making it a promising approach for secure networked control systems. 
Most existing results pre-design a controller in the plaintext domain and then implement it over encrypted data. 
However, this can be problematic because HE induces non-negligible communication and computation delays that typically increase with the security level, potentially degrading control performance and even destabilizing the closed-loop system. 
To address this issue, we propose a co-design framework for cryptographic parameters and delay-aware feedback gain. 
We first characterize an upper bound of the encryption-induced delay as a function of the cryptographic parameters. 
Then, for a given set of cryptographic parameters and feedback gain, we derive a sufficient condition under which the closed-loop system remains stable for all admissible delays, expressed as a finite set of linear matrix inequalities.
This leads to a tractable outer-inner design procedure: the outer loop selects cryptographic parameters satisfying the desired security level, while the inner loop seeks a stabilizing delay-aware feedback gain.
}

\keywords{%
Encrypted control, homomorphic encryption, security, privacy
}

\begin{document}

\maketitle

%-----------------------------------------------------------------------
\section{Introduction}

Networked control systems allow resource-limited devices to efficiently outsource computationally intensive tasks to remote servers \cite{ZhanHanq19}. 
However, this architecture requires the transmission of data that may contain sensitive information, giving rise to privacy concerns \cite{AnanNguy25}. 
In particular, communication channels are vulnerable to eavesdropping, and the remote servers may be semi-honest, i.e., they may correctly execute the prescribed protocol while attempting to infer private information from the received data. 
The leaked information can be exploited to design cyberattacks capable of causing significant physical damage, as demonstrated by the incidents in \cite{HemsFish18,Slow19}. 

Recently, encrypted control has emerged as a promising framework for enhancing the security of networked control systems \cite{KogiFuji15,KimjKimd22,SchlBinf23}.
By employing homomorphic encryption (HE), which enables arithmetic operations to be performed directly on encrypted data without decryption, it protects both the communication and computation stages. 
More specifically, sensor measurements are
encrypted and transmitted to an untrusted remote server, where the control input is computed over encrypted data and sent back to the actuator for decryption (see Fig.~\ref{fig:diagram}). 

Most related studies first design a controller in the plaintext domain and then integrate HE, while preserving the performance of the unencrypted controller \cite{KimjShim23,TeraSada24,SchlAllg24}. 
However, they often neglect the effects of \textit{encryption-induced delays}. 
Such delays arise from, for example, the transmission of ciphertexts whose sizes are typically much larger than that of plaintexts, as well as the computational burden of encryption, homomorphic evaluation, and decryption. 
These may cause significant performance degradation and even destabilize the closed-loop system \cite{Cloo08}.
Although \cite{TeraSada23} presented a method for jointly designing the key length of a cryptosystem and control laws achieving a desired security level and control performance, it has not accounted for encryption-induced delays.

\begin{figure}[t]
\centering
\begin{tikzpicture}[font=\small]

% =========================================================
% Top row
% =========================================================
\node[block] (q1) {ZOH};
\node[block, right=1cm of q1] (plant) {Plant};
\node[block, right=1cm of plant, minimum width=4.5em, minimum height=2.6em] (q2) {};

% internal switch drawing
\draw ($(q2.west)+(0.22cm,0.1cm)$) -- ($(q2.center)+(-0.18cm,0.1cm)$);
\draw ($(q2.center)+(-0.05cm,0.3cm)$) -- ($(q2.center)+(0.22cm,0.1cm)$);
\fill ($(q2.center)+(-0.05cm,0.3cm)$) circle (1pt);
\fill ($(q2.center)+(-0.18cm,0.1cm)$) circle (1pt);
\draw ($(q2.center)+(0.22cm,0.1cm)$) -- ($(q2.east)+(-0.22cm,0.1cm)$);

\node at ($(q2.center)+(0,-0.2cm)$) {Sample};
\node at ($(q2.center)+(0.52cm,0.26cm)$) {\scriptsize $T_s$};
% =========================================================
% Middle row
% =========================================================
\node[block, below=0.7cm of q1] (scale) {Scale};
\node[block, below=0.7cm of q2] (quant) {Quantize};

\node[block, below=0.7cm of scale] (dec) {$\Dec$};
\node[block, below=0.7cm of quant] (enc) {$\Enc$};

% =========================================================
% Controller block
% Reduced vertical distance from plant
% =========================================================
\node[block, below=4.3cm of plant, minimum width=5.5em] (ctrl) {Encrypted\\Controller};

% =========================================================
% Top chain
% =========================================================
\draw[line] (q1) -- node[midway, above] {$u(t)$} (plant);
\draw[line] (plant) -- node[midway, above] {$x(t)$} (q2);

% Left branch
\draw[line, dashed] (dec) -- node[midway, right] {$\Dec(\bfu[k])$} (scale);
\draw[line, dashed] (scale) -- node[midway, right] {$u_e[k]$} (q1);

% Right branch
\draw[line, dashed] (q2) -- node[midway, right] {$x[k]$} (quant);
\draw[line, dashed] (quant) -- node[midway, right] {$\left\lceil \frac{x[k]}{\sfr} \right\rfloor$} (enc);

\draw[line, dashed] (enc) |- 
    node[midway, above, xshift=-6mm] {$\bfx[k]$}
    (ctrl);

% =========================================================
% Controller interconnection
% =========================================================
\draw[line, dashed] (ctrl.west) -| 
    node[midway, right, yshift=8pt, xshift=9pt] {$\bfu[k]$}
    (dec.south);

% =========================================================
% Background group boxes
% =========================================================
\begin{pgfonlayer}{background}

% Top group
\node[
    draw,
    dotted,
    thick,
    rounded corners,
    fit=(q1)(plant)(q2)(dec)(enc),
    inner sep=8pt
] (boxTop) {};

% Bottom group
\node[inner sep=0pt, fit=(ctrl)] (boxRestFit) {};

% Make network box horizontally match the plant-side box
\coordinate (cyberSW) at ($(boxTop.south west |- boxRestFit.south west)+(0pt,-12pt)$);
\coordinate (cyberNE) at ($(boxTop.south east |- boxRestFit.north east)+(0pt,12pt)$);

\node[
    draw,
    dotted,
    thick,
    rounded corners,
    fit=(cyberSW)(cyberNE),
    inner sep=0pt
] (boxCyber) {};

% Labels
\node[anchor=south west]
    at ($(boxTop.north west)+(-9pt,0pt)$)
    {\bf Plant side};

\node[anchor=south west]
    at ($(boxCyber.north west)+(-9pt,0pt)$)
    {\bf Network};

\end{pgfonlayer}

\node[anchor=north, inner sep=3pt] (legend)
    at ($(boxCyber.south)+(0,-2mm)$) {%
    \begin{tabular}{@{}l l l l@{}}
    \tikz{\draw[line] (0,0)--(7mm,0);} & Continuous signal &\qquad
    \tikz{\draw[line, dashed] (0,0)--(7mm,0);} & Discrete signal
    \end{tabular}
};

\end{tikzpicture}
\caption{Configuration of encrypted control system.}
\label{fig:diagram}
\end{figure}
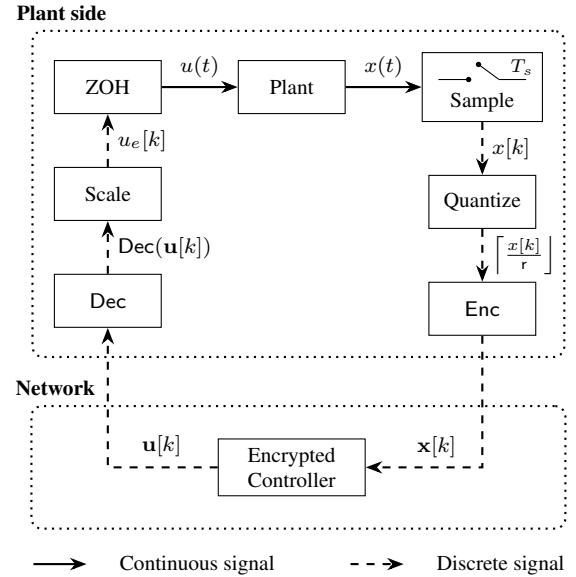

In this paper, we propose a co-design framework for cryptographic parameters and delay-aware feedback gain for encrypted control systems.
To this end, we first characterize an upper bound of the encryption-induced delay as a function of the cryptographic parameters.
Then, for a fixed set of cryptographic parameters and feedback gain, we derive a sufficient condition under which the closed-loop system remains stable for all admissible delays, expressed as a finite set of linear matrix inequalities (LMIs).
This leads to a tractable outer-inner design procedure that searches over a candidate set to identify cryptographic parameters satisfying a desired security level and, for each such parameter set, seeks a stabilizing feedback gain.

\noindent\textbf{Notation:}
Let $\bbZ$, $\bbZ_{\ge 0}$, $\bbN$, and $\bbR$ denote the sets of integers, nonnegative integers, positive integers, and real numbers, respectively.

%-----------------------------------------------------------------------
\section{Preliminaries and Problem Formulation}

\subsection{LWE based cryptosystem}\label{subsec:enc}

This section briefly reviews the LWE based cryptosystem of \cite{Rege09,ChilGama16}, focusing on its multiplicative homomorphic property and computational complexity.

For a prime number $q\in\bbN$, we consider $\bbZ_q:=\bbZ\cap [-q/2,q/2)$ as the space of plaintexts. 
The modulo operation that maps $\bbZ$ onto $\bbZ_q$ is denoted by
\begin{align*}
    a\modp q:=a - \left\lfloor \frac{a+q/2}{q}\right\rfloor q\in\bbZ_q, ~~~~\forall a\in\bbZ,
\end{align*}
which applies entrywise to vectors and matrices. 
The secret key $\sk\in\bbZ_q^N$, where $N\in\bbN$ is a power-of-two, is generated by drawing each component uniformly at random from the ternary set $\{-1,0,1\}$.
Only an entity with access to it can perform both encryption and decryption.

The LWE based scheme supports matrix-vector multiplication to be evaluated over encrypted data.
Specifically, it employs two distinct encryption algorithms, denoted by
$\Enc'$ and $\Enc$, for the multiplier and multiplicand, respectively:
\begin{align*}
    \Enc'&:\bbZ_q \to \bbZ_q^{(N+1) \times d(N+1)}=:\calC', \\
    \Enc&:\bbZ_q \to \bbZ_q^{N+1}=:\calC,
\end{align*}
where $d := \lceil \log_2 q \rceil$.
The corresponding decryption algorithm is denoted by $\Dec:\calC\to\bbZ_q$. 
With a slight abuse of notation, $\Enc'$, $\Enc$, and $\Dec$ are applied entrywise to vectors and matrices.
Then, given $M\in\bbZ_q^{h\times l}$ and $m\in\bbZ_q^l$, there exists a homomorphic multiplication operation\footnote{For simplicity, we omit the dependence of $\otimes$ on $h$ and $l$.} $\otimes:(\calC')^{h \times l}\times \calC^{l} \to \calC^h$, referred to as the ``external product,'' such that \cite{ChilGama16} 
\begin{align}\label{eq:homo}
    \Dec(\Enc'(M)\otimes \Enc(m)) \approx Mm \modp q.
\end{align}

The approximation error in \eqref{eq:homo} is caused by the ``errors'' injected during encryption, which are fundamental to the security of the LWE based scheme.
These errors are typically drawn from a zero-mean discrete Gaussian distribution with standard deviation $\sigma>0$. 
We refer to $\theta=(N,q,\sigma)$ as the cryptographic parameter set and denote the corresponding security level by $\lambda(\theta)$.
In general, the security level increases with $N$ and $\sigma/q$ \cite{Albr21}, and can be estimated using a tool called the LWE estimator \cite{AlbrPlay15}.

For the scalar case, i.e., when $h=l=1$, the computational complexities of $\Enc'$, $\Enc$, $\Dec$, and $\otimes$ are summarized in Table~\ref{tab:complexity} using the big-O notation.
It can be seen that the complexities primarily depend on the parameters $N$ and $q$, indicating an inherent tradeoff between security level and computational complexity.

\begin{table}[t]
\centering
\caption{Asymptotic computational complexity of cryptographic operations in the scalar case ($h=l=1$).}
\label{tab:complexity}
\setlength{\tabcolsep}{6pt}
\renewcommand{\arraystretch}{1.1}
\small
\begin{tabular}{lcccc}
\hline\hline
Operation  & $\Enc'$ & $\Enc$ & $\Dec$ & $\otimes$ \\
\hline
Complexity & $\calO(N^{2}d)$ & $\calO(Nd)$ & $\calO(Nd)$ & $\calO(Nd\log N)$ \\
\hline\hline
\end{tabular}
\end{table}

% \begin{table}[t]
% \centering
% \caption{Asymptotic computational complexity of cryptographic operations for the scalar case (\(h=l=1\)).}
% \label{tab:complexity}
% \setlength{\tabcolsep}{4pt}
% \small
% \begin{tabular}{c c c c c}
% \hline
% Operation & $\Enc'$ & $\Enc$ & $\Dec$ & $\otimes$ \\ \hline
% Complexity & $\calO(N^{2}d)$ & $\calO(Nd)$ & $\calO(Nd)$ & $\calO(Nd\log N)$ \\ \hline
% \end{tabular}
% \end{table}

\subsection{Problem formulation}\label{subsec:prob}

Consider a continuous-time linear plant written by
\begin{equation}
\dot{x}(t)=Ax(t)+Bu(t),~~~~ x(0)=x_{\ini},
\label{eq:plant}
\end{equation}
where $x(t)\in\bbR^n$ is the state with the initial value $x_\ini\in\bbR^n$ and $u(t)\in\bbR^m$ is the control input.
To stabilize \eqref{eq:plant}, suppose that a discrete-time static state feedback controller has been designed as 
\begin{align}\label{eq:controller}
    u[k] = K x[k],
\end{align}
where the state $x[k]:=x(k T_s)$ is obtained with a sampling time of $T_s>0$ for $k\in\bbZ_{\ge 0}$, and $K\in\bbR^{m \times n}$ is the feedback gain. 
Ideally, the signal $u[k]$ is fed-back to the plant \eqref{eq:plant} through a zero-order hold (ZOH), as 
% $u(t)=u[k]$ for $k T_s \le t < (k+1) T_s$.
\begin{align*}
   u(t)=u[k],~~~~ k T_s \le t < (k+1) T_s. 
\end{align*}

Now, we consider a networked control system implementation of \eqref{eq:plant} and \eqref{eq:controller} in which the control input is computed by a remote server residing in the network layer.
In particular, we consider a clock-driven sensor that samples the state periodically, whereas the controller/actuator are event-driven and computes/applies the control input as soon as required data are available.
To protect the feedback gain and sampled state against eavesdropping on the communication channels and the server, we integrate the LWE based scheme introduced in Section~\ref{subsec:enc}. 
The overall configuration is illustrated in Fig.~\ref{fig:diagram}.

The implementation details are as follows.
During the offline phase, the gain $K$ is quantized using a scale factor $1/\sfs\ge 1$, encrypted as
\begin{align}\label{eq:gainEnc}
    \bfK = \Enc'\left( \left\lceil \frac{K}{\sfs} \right\rfloor \modp q \right) \in \left( \calC' \right) ^{m\times n},
\end{align}
and stored at the server.
At each sampling instant $t=k T_s$ of the online phase, the sensor quantizes the sampled state with a resolution $\sfr>0$, encrypts it, and transmits 
\begin{align}\label{eq:stateEnc}
    \bfx[k] := \Enc\left(\left\lceil \frac{x[k]}{\sfr} \right\rfloor \modp q\right) \in \calC^{n}
\end{align}
to the server.
The server then evaluates $\bfu[k] = \bfK \otimes \bfx[k]$
and transmits it back to the actuator, where it is decrypted and rescaled to obtain 
\begin{align}\label{eq:actuator}
    u_e[k] := \sfr\sfs\cdot \Dec(\bfu[k])\approx \sfr\sfs \cdot   \left\lceil \frac{K}{\sfs} \right\rfloor \left\lceil \frac{x[k]}{\sfr} \right\rfloor \modp q,
\end{align}
where the approximation error is due to \eqref{eq:homo}.

Indeed, $u_e[k]$ can be made arbitrarily close to $u[k]$ by choosing $1/\sfr$ and $1/\sfs$ sufficiently large, thereby reducing the precision loss incurred by the rounding operations, and by choosing $q$ sufficiently large to prevent overflow under the modulo operation.
Since the focus of this paper is not on analyzing the effects of quantization or encryption errors, we henceforth make the simplifying assumption that 
\begin{align*}
    u_e[k]\equiv u[k].
\end{align*}
Analyses on the effects of such errors on the closed-loop stability can be found in \cite[Theorem~2]{KimjShim23}
or \cite[Theorem~1]{JangLeej25}.

Recall that the gain $K$ has been designed for the ideal delay-free case, i.e., when $u[k]$ can be applied exactly at $kT_s$.
In the encrypted implementation, however, non-negligible delays arise from the transmission of ciphertexts whose message sizes are substantially larger than that of plaintexts, as well as from the time required for encryption, decryption, and homomorphic multiplication.

Formally, let $\tau[k]\in\bbR$ denote the total time elapsed from the sampling of $x[k]$ to the computation of $u_e[k]$ at the actuator, which we refer to as the \textit{total delay} at sampling instant $k$.
For analytical simplicity, we impose the following assumption.

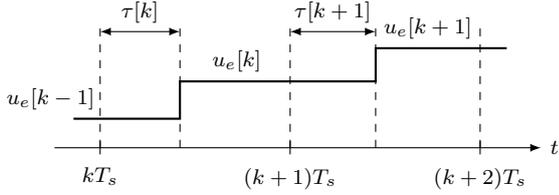
\begin{figure}[t]
\centering
\begin{tikzpicture}[x=1cm,y=1.1cm,>=latex,font=\small]

% spacing between sampling instants
\def\Ts{2.5}
\def\tauk{0.42}
\def\tauknext{0.45}

% actual delayed instants
\pgfmathsetmacro{\tkavail}{\tauk*\Ts}
\pgfmathsetmacro{\tknextavail}{\Ts+\tauknext*\Ts}

% time axis
\draw[->] (-0.6,0) -- ({2*\Ts+0.8},0) node[right] {$t$};

% sampling instants
\foreach \x/\lab in {
  0/{$kT_s$},
  \Ts/{$(k+1)T_s$},
  {2*\Ts}/{$(k+2)T_s$}
}{
  \draw[dashed] (\x,0) -- (\x,1.45);
  \draw (\x,0.08) -- (\x,-0.08);
  \node[below=4pt] at (\x,0) {\lab};
}

% delayed availability instants
\draw[dashed] (\tkavail,0) -- (\tkavail,1.45);
\draw[dashed] (\tknextavail,0) -- (\tknextavail,1.45);

% latency arrows
\draw[<->] (0,1.40) -- node[above] {$\tau_\theta$} (\tkavail,1.40);
\draw[<->] (\Ts,1.40) -- node[above] {$\tau_\theta$} (\tknextavail,1.40);

% sampled control held by ZOH
\draw[thick]
  (-0.35,0.35) -- (\tkavail,0.35)
  -- (\tkavail,0.80) -- (\tknextavail,0.80)
  -- (\tknextavail,1.20) -- ({2*\Ts+0.35},1.20);

% labels for active inputs
\node[above] at (-0.65,0.35) {$u_e[k-1]$};
\node[above] at (1.8,0.80) {$u_e[k]$};
\node[above] at ({0.5*(\tknextavail+2*\Ts)},1.20) {$u_e[k+1]$};

% availability markers
% \draw[->] (\tkavail,1.33) -- (\tkavail,0.84);
% \draw[->] (\tknextavail,1.33) -- (\tknextavail,1.24);

\end{tikzpicture}
\caption{Timing diagram of the zero-order held input $u(t)$ under the total delay $\tau[k]=\tau_\theta\in[0,T_s)$.}
\label{fig:delay}
\end{figure}

\begin{assumption}\label{asm:delay}
The total delay is constant across sampling instants, i.e., $\tau[k] = \tau_{\theta}$ for all $k\in\bbZ_{\ge 0}$, where $\tau_\theta>0$ is a constant determined by the cryptographic parameter set $\theta$ and is not assumed to be known a priori.
\end{assumption}

When $\tau_\theta \in[0,T_s)$, the actuator holds the previous input until the current one becomes available, as illustrated in Fig.~\ref{fig:delay}.
Specifically, the control input is applied as
\begin{align}\label{eq:delayInput}
    u(t) = 
    \begin{cases}
        u_e[k-1], & k T_s\le t <k T_s + \tau_\theta, \\
        u_e[k], & kT_s+\tau_\theta \le t <(k+1) T_s.
    \end{cases}
\end{align}

It is well known that such an input delay can degrade performance and even destabilize the closed-loop system \cite{ZhanBran01,Cloo08}.
Motivated by this issue, we propose a co-design framework for the cryptographic parameter set $\theta $ and the feedback gain $K$.
Specifically, our goal is to jointly determine $(\theta,K)$ such that the following requirements are satisfied:
\begin{enumerate}[label=\textbf{(R\arabic*)},ref=R\arabic*,leftmargin=*,align=left]
    \item \label{req:security} \textit{Security requirement}: For a desired security level $\lambda^*$, the condition $\lambda(\theta)\ge \lambda^*$ holds.
    \item \label{req:delay} \textit{Delay requirement}: The total delay satisfies $\tau_\theta\in[0,T_s)$.
    \item \label{req:stability} \textit{Stability requirement}: The closed-loop system \eqref{eq:plant} with the delayed input \eqref{eq:delayInput} is stable.
\end{enumerate}

As a first step toward developing this co-design framework, we focus on static state-feedback controllers of the form \eqref{eq:controller}.
We also note that controller design methods under time-varying delays or delays exceeding the sampling time have been studied in the literature \cite{Cloo08,LiukSeli19}. 
Extending the proposed framework to incorporate such methods is left for future work.

%-----------------------------------------------------------------------
\section{Main Results}

\subsection{Quantification of the total delay}

We begin by quantifying the total delay $\tau_\theta$ as a function of the cryptographic parameter set $\theta$.
Under Assumption~\ref{asm:delay}, we decompose the total delay as
\begin{equation}
    \tau_{\theta}
    =
    \tau_{\theta}^{\Comm}
    +\tau_{\theta}^{\Enc}
    +\tau_{\theta}^{\Dec} 
    +\tau_{\theta}^{\otimes},
    \label{eq:tau_decomp}
\end{equation}
where the terms on the right-hand-side, in order, represent the delays due to communication, sensor-side encryption, actuator-side decryption, and server-side homomorphic multiplication.
The time required for the encryption and transmission of $\bfK$ are excluded from \eqref{eq:tau_decomp}, since they are performed during the offline phase.
Note that other types of delays, such as queuing or buffering delays, can be incorporated into $\tau_\theta$ without loss of generality.

We adopt the following standard communication model \cite{LiukSeli19,Kuro13} to obtain an explicit bound on the communication delay.

\begin{assumption}\label{asm:comm}
The communication link between the plant and the server has a fixed transmission rate $R>0$ and a fixed propagation delay $\delta\ge 0$. 
\end{assumption}

Here, $R$ denotes the number of bits transmitted per unit time, measured in bits per second, and $\delta$ denotes the propagation delay measured in seconds, which depends primarily on the distance between the plant and the server.
Accordingly, transmitting a message of length $L$ bits requires at most $L/R + \delta$ seconds.
Based on this model, the following lemma characterizes $\tau_\theta^{\Comm}$.

\begin{lemma}\label{lem:comm}
Under Assumption~\ref{asm:comm}, the communication delay is upper bounded as
\begin{equation}\label{eq:propCommToShow}
    \tau_\theta^{\Comm} \le \frac{(n+m)(N+1)d}{R}+2\delta.
\end{equation}
\end{lemma}

\begin{proof}
The communication delay is caused by the transmission of $\bfx[k]\in\calC^n$ and $\bfu[k]\in\calC^m$.
Since each ciphertext in $\calC=\bbZ_q^{N+1}$ contains $N+1$ entries in $\bbZ_q$, and each entry requires at most $d=\lceil\log_2 q\rceil$ bits, transmitting $\bfx[k]$ and $\bfu[k]$ requires at most
\begin{align*}
    \frac{n(N+1)d}{R}+\delta
    \qquad\text{and}\qquad
    \frac{m(N+1)d}{R}+\delta
\end{align*}
seconds, respectively.
Summing these two terms yields \eqref{eq:propCommToShow}, and this concludes the proof.
\end{proof}

We next quantify the remaining components of $\tau_\theta$, namely, the delays due to encryption, decryption, and homomorphic multiplication.
The following lemma provides an explicit upper bound.

\begin{lemma}\label{lem:comp_delay}
There exist positive constants $c_{\Enc}$, $c_{\Dec}$, and $c_{\otimes}$ such that
\begin{multline}\label{eq:comp_delay_bound}
    \tau_{\theta}^{\Enc}
    +\tau_{\theta}^{\Dec}
    +\tau_{\theta}^{\otimes}
    \\ \le
    Nd \left( c_{\Enc}n
    +c_{\Dec}m
    +c_{\otimes}mn\,\log N \right).
\end{multline}
\end{lemma}
\begin{proof}
By Table~\ref{tab:complexity}, there exist positive constants $c_{\Enc}$, $c_{\Dec}$, and $c_{\otimes}$, independent of $\theta$, such that the execution times of one scalar encryption, decryption, and homomorphic multiplication are upper bounded by $c_{\Enc}Nd$, $c_{\Dec}Nd$, and $c_{\otimes}dN\log N$, respectively.
Since $x[k]\in\bbR^n$ and $\bfu[k]\in\calC^m$ are encrypted and decrypted entrywise, respectively, we have
\begin{equation*}
    \tau_{\theta}^{\Enc}\le c_{\Enc}nNd,
    \qquad
    \tau_{\theta}^{\Dec}\le c_{\Dec}mNd.
\end{equation*}
Moreover, computing $\bfK\otimes\bfx[k]$ requires $mn$ scalar homomorphic multiplications, and hence
\begin{equation*}
    \tau_{\theta}^{\otimes}\le c_{\otimes}mn\,Nd\log N.
\end{equation*}
Summing the above inequalities yields \eqref{eq:comp_delay_bound}, and this concludes the proof.
\end{proof}

The constants $c_{\Enc}$, $c_{\Dec}$, and $c_{\otimes}$ are implementation dependent, and vary with hardware specification, choice of HE library, and specific algorithms used to realize each cryptographic operation.
A practical strategy is to execute each operation repeatedly and use the worst-case as an empirical estimate of those constants.

Although the total delay $\tau_\theta$ is unknown under Assumption~\ref{asm:delay}, Lemmas~\ref{lem:comm} and \ref{lem:comp_delay} provide the following upper bound as a function of the cryptographic parameter $\theta$:
\begin{multline}\label{eq:delayIneq}
    \tau_\theta \le \frac{(n+m)(N+1)d}{R}+2\delta \\ + Nd \big( c_{\Enc}n
    +c_{\Dec}m
    +c_{\otimes}mn\,\log N \big) =:\bar{\tau}_\theta.
\end{multline}

\subsection{Delay-aware gain design}

In what follows, for a fixed cryptographic parameter $\theta$ and feedback gain $K$, we derive a sampled-data representation of the closed-loop system and analyze its stability.

By exactly discretizing the closed-loop system \eqref{eq:plant} under the delayed input \eqref{eq:delayInput} at the sampling instants and using the relation $u_e[k]=u[k]=Kx[k]$, we obtain
\begin{align}\label{eq:loop}
    x[k+1] & = \Phi x[k]  \!+\! \Gamma_0(\tau_\theta) u_e[k] \!+\! \Gamma_{1}(\tau_\theta) u_e[k-1] \\ 
    &= \left(\Phi + \Gamma_0(\tau_\theta)  K \right) x[k] \!+\! \Gamma_{1}(\tau_\theta) K x[k-1], \nonumber
\end{align}
where $\Phi = \exp(AT_s)$ and 
\begin{align*}
     \Gamma_0(\tau_\theta) &:= 
    \int_{\tau_\theta}^{T_s} \exp(A(T_s-h))\,dh\, B, \\
    \Gamma_1(\tau_\theta) &:= 
    \int_{0}^{\tau_\theta} \exp(A(T_s-h))\,dh\,B.
\end{align*}
By defining the augmented state $\xi[k]:=[x[k]^\top,x[k-1]^\top]^\top \in \bbR^{2n}$, we further obtain
\begin{equation}
    \xi[k+1]
    =
    \bar A(\tau_\theta,K)\,\xi[k],
    \label{eq:augmented_system_static}
\end{equation}
where
\begin{equation}
    \bar A(\tau_\theta,K)
    :=
    \begin{bmatrix}
        \Phi+\Gamma_0(\tau_\theta)K & \Gamma_1(\tau_\theta)K\\
        I_n & 0_{n\times n}
    \end{bmatrix}.
    \label{eq:Abar_static}
\end{equation}

\begin{algorithm}[t]
\caption{Co-design algorithm}
\label{alg:outer_inner}
\begin{algorithmic}[1]
\renewcommand{\algorithmicrequire}{\textbf{Input:}}
\Require Model parameters $(A,B)$, desired security level $\lambda^*$, and sampling time $T_s$
\renewcommand{\algorithmicrequire}{\textbf{Output:}}
\Require A pair $(\theta,K)$ satisfying (R1)--(R3)
\State Construct a candidate set $\Theta$ of cryptographic parameters such that $\lambda(\theta)\ge \lambda^*$ for all $\theta \in \Theta$
\For{$\theta\in\Theta$}
    \State Compute $\bar{\tau}_\theta$ in \eqref{eq:delayIneq}
    \If{$\bar{\tau}_\theta\ge T_s$}
        \State \textbf{Continue}
    \Else
        \State Construct a finite candidate set of gains $\calK$
    \For{$K\in\calK$}
        \If{\eqref{eq:LMI} is feasible for some $P\succ 0$}
            \State \textbf{Return} $(\theta,K)$
        \EndIf
    \EndFor
    \EndIf
\EndFor
\end{algorithmic}
\end{algorithm}

Consequently, the stability requirement (R3) reduces to verifying the Schur stability of all matrices in the set
\begin{align*}
    \Psi(\theta,K) := \{\bar A(\tau_\theta,K) \mid \tau_\theta\in[0,\bar\tau_\theta]\}.
\end{align*}
Since this requires checking infinitely many matrices, we derive a tractable sufficient condition in the sequel, expressed in terms of a finite set of LMIs.

By the Jordan form representation of $A$, the matrix $\Gamma_1(\tau_\theta)$ can be expressed as \cite[Appendix~B.2]{Cloo08}
\begin{equation}
    \Gamma_1(\tau_\theta)=\textstyle\sum_{i=1}^{\nu}\alpha_i(\tau_\theta)G_i,
    \qquad \forall\tau_\theta\in[0,\bar\tau_\theta]
    \label{eq:Gamma1_decomp}
\end{equation}
for some $\nu\in\bbN$,
where \(\alpha_i:[0,\bar\tau_\theta]\to\mathbb{R}\) are continuous scalar functions and
\(G_i\in\bbR^{n\times m}\) are constant matrices for $i=1,\ldots, \nu$.
Then, with $\Gamma:=\Gamma_0(\tau_\theta)+\Gamma_1(\tau_\theta)
    =\int_{0}^{T_s} e^{A(T_s-h)}dhB$,
which is independent of $\tau_\theta$, we can decompose $\bar{A}(\tau_\theta,K)$ as
\begin{equation*}
    \bar A(\tau_\theta,K)
    =
    \underbrace{
    \begin{bmatrix}
        \Phi+\Gamma K & 0_{n \times n} \\
        I_n & 0_{n \times n}
    \end{bmatrix}}_{=:M_0(K)}
    +\textstyle\sum_{i=1}^{\nu}\alpha_i(\tau_\theta)M_i(K),
\end{equation*}
where
\begin{equation}
    M_i(K):=
    \begin{bmatrix}
        -G_iK & G_iK\\
        0_{n \times n} & 0_{n \times n}
    \end{bmatrix},
    \qquad i=1,\dots,\nu.
    \label{eq:Mi_def}
\end{equation}

For each \(i=1,\dots,\nu\), let
\begin{equation}
    \alpha_i^-:=\min_{\tau_\theta\in[0,\bar\tau_\theta]}\alpha_i(\tau_\theta),
    ~~~~ \alpha_i^+:=\max_{\tau_\theta\in[0,\bar\tau_\theta]}\alpha_i(\tau_\theta).
    \label{eq:alpha_bounds}
\end{equation}
In fact, the functions $\alpha_i(\cdot)$ can be computed explicitly \cite[Appendix~B.2]{Cloo08}, and thus, \eqref{eq:alpha_bounds} can be obtained by solving a one-dimensional optimization problem.
Then, for each \(\eta=(\eta_1,\dots,\eta_\nu)\in\{-,+\}^\nu\), we define
\begin{equation}
    M_\eta(\theta,K):=
    M_0(K)+\textstyle\sum_{i=1}^{\nu}\alpha_i^{\eta_i} M_i(K).
    \label{eq:Asigma_def}
\end{equation}

\begin{proposition}\label{prop:stab}
Suppose that a cryptographic parameter set $\theta$ and a feedback gain $K\in\bbR^{m \times n}$ are given.
Under Assumptions~\ref{asm:delay} and~\ref{asm:comm},
if there exists a symmetric matrix \(P\succ0\) such that
\begin{equation}\label{eq:LMI}
    M_\eta(\theta,K)^\top PM_\eta(\theta,K)-P\prec0,
    ~~~~ \forall \eta\in\{-,+\}^{\nu},
\end{equation}
then every matrix in $\Psi(\theta,K)$ is Schur stable. 
\end{proposition}

\begin{proof}
By definition, \(\alpha_i(\tau_\theta)\in[\alpha_i^-,\alpha_i^+]\) for all
\(i=1,\ldots,\nu\). 
Hence, every matrix in \(\Psi(\theta,K)\) belongs to
the convex hull of the matrices \(M_\eta(\theta,K)\), i.e.,
\begin{equation}
    \Psi(\theta,K)\subseteq
    {\mathrm{co}}\{M_\eta(\theta,K)\mid \eta\in\{-,+\}^\nu\}.
    \label{eq:Psi_overapprox}
\end{equation}
The conclusion then follows from the
standard common quadratic Lyapunov argument for convex polytopes; see the proof
of \cite[Theorem~4.3.1]{Cloo08}. 
\end{proof}

The condition in Proposition~\ref{prop:stab} can be tightened whenever a tighter range of $\tau_\theta$ is given.
For example, if a lower bound $\underline{\tau}_\theta>0$ is known, recomputing \eqref{eq:alpha_bounds} over $[\underline{\tau}_\theta, \bar{\tau}_\theta]$ reduces the conservativeness of \eqref{eq:LMI}.

\subsection{Co-design of cryptographic parameters and gain}

The proposed co-design method for the cryptographic parameter set $\theta$ and the
feedback gain $K$ is formulated as an outer--inner procedure, summarized in
Algorithm~\ref{alg:outer_inner}. 

In the outer loop, the cryptographic parameter set $\theta$ is selected from a discrete candidate set, so as to satisfy the security and the delay requirements (R1) and (R2).
The candidate set is discrete because the secret key length $N$ and the modulus $q$ are restricted to powers-of-two and prime numbers, respectively; in practice, such candidates are often precomputed in cryptographic libraries or can be chosen according to standard guidelines \cite{Albr21}.

For each fixed candidate $\theta$, the inner loop checks the feasibility of the LMIs in \eqref{eq:LMI} for each feedback gain in a finite candidate set $\calK$. 
The set $\calK$ is introduced for computational tractability and can, for example, be chosen as the set of gains designed for the delay-free case.
Consequently, any pair $(\theta,K)$ returned by the algorithm satisfies the requirements (R1)--(R3). In the present work, we adopt this feasibility based design for simplicity. 
If an additional performance criterion is imposed, the same framework can be extended to select a more desirable gain among the feasible candidates.

Regarding practical applicability, the LMIs in \eqref{eq:LMI} can be efficiently checked using standard semidefinite programming solvers given $K$. 
In principle, one may also attempt to optimize $K$ and $P$ simultaneously. 
However, the resulting matrix inequalities become nonconvex, and thus, developing suitable relaxations or alternative reformulations that enable a tractable simultaneous search over $K$ and $P$ is left for future work.

\section{Simulation Results}

This section provides simulation results to demonstrate the effectiveness of the proposed method through a numerical example. 
Consider the plant \eqref{eq:plant} given as
\begin{align*}
    A = 
    \begin{bmatrix}
        0 & 1 \\
        0 & 0 
    \end{bmatrix}, ~~~~
    B =
    \begin{bmatrix}
        0 \\
        1
    \end{bmatrix}
\end{align*}
with the initial condition $x_\ini = [1, 0]^\top$ and sampling time $T_s=\SI{25}{\milli\second}$.
We empirically set the communication transmission rate to $R=\SI{250}{\mega\bit/\second}$ and the propagation delay to $\delta=\SI{2}{\milli\second}$.

The encrypted controller was implemented using the Lattigo library \cite{Lattigo} with the Brakerski-Gentry-Vaikuntanathan (BGV) variant \cite{BrakGent14} of the LWE based scheme.
The target security level was set to $\lambda^*=128$ bits \cite{Albr21}, and three candidate cryptographic parameter sets satisfying the security requirement were selected, which are specified in Table~\ref{tab:bench}.
For each candidate parameter set, the upper bound on the computation delay was estimated as the worst-case value among $100$ runtime measurements.
Then, $\bar{\tau}_\theta$ was obtained by adding this worst-case computation delay to the upper bound on the communication delay, as in \eqref{eq:delayIneq}.

Among the three candidates, it can be found from Table~\ref{tab:bench} that only $\theta_1$ satisfies~(R2). 
Hence, for $\theta_1$, we designed two candidate feedback gains $K_1=[-960, -50]$ and $K_2=[-72, -16.1]$
via pole placement for the delay-free case. 
Precisely, $K_1$ and $K_2$ were chosen so that the matrices $\bar{A}(0,K_1)$ and $\bar{A}(0,K_2)$ have poles at $(0.20,0.25)$ and $(0.775,0.800)$, respectively, thus stabilizing the delay-free closed-loop system. 
We then tested the LMIs in \eqref{eq:LMI} using \texttt{cvxpy}. While no feasible solution was found for $K_1$, the LMIs were feasible for $K_2$ with one feasible solution given by
\begin{equation*}
P =
\begin{bmatrix}
 0.740 &  0.073 & -0.101 & -0.019\\
 0.073 &  0.052 & -0.004 & -0.032\\
-0.101 & -0.004 &  0.176 &  0.013\\
-0.019 & -0.032 &  0.013 &  0.032
\end{bmatrix}.
\end{equation*}

\begin{table}[t]
\centering
\caption{Candidate parameter sets satisfying (R1) and corresponding empirical upper bounds on total delay.}
\label{tab:bench}
\begin{tabular}{ccccc}
\hline
 & $\log_2 N$ & $\log_2q$ & $\sigma$ & $\bar{\tau}_\theta$ [\SI{}{\milli\second}] \\
\hline
\hline
$\theta_1$ & 12 & 109 & 3.2 & 19.150   \\
$\theta_2$ & 13 & 218 & 3.2 & 80.810 \\
$\theta_3$ & 14 & 438 & 3.2 & 304.225 \\
\hline
\end{tabular}
\end{table}

We evaluated the closed-loop performance using the gains $K_1$ and $K_2$ for four different values of total delay $\tau_\theta$; Fig.~\ref{fig:result} shows the corresponding state norm $\|x(t)\|$. 
The results show that $K_2$ stabilizes the closed-loop system for all considered total delays, whereas the performance of the delay-unaware gain $K_1$ deteriorates as $\tau_\theta$ increases. 
In particular, the closed-loop system becomes unstable when $\tau_\theta = 3\bar{\tau}_\theta/4$. 
These results demonstrate both the effectiveness and necessity of the proposed co-design framework.

\begin{figure}[t]
\centering
\input{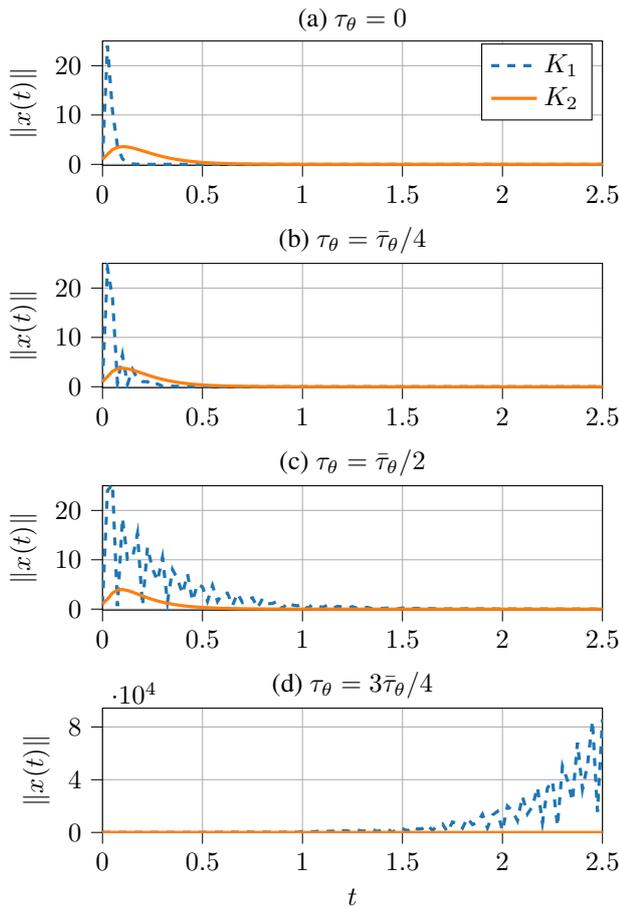}
\vspace{-0.5cm}
\caption{Closed-loop state norm $\|x(t)\|$ using the gains $K_1$ and $K_2$ for four values of the total delay $\tau_\theta$.}
\label{fig:result}
\end{figure}

%-----------------------------------------------------------------------

\section*{Acknowledgment}
This work was supported by the National Research Foundation of Korea (NRF) grant funded by the Korea government (MSIT) (No. RS-2024-00353032 and RS-2026-25504174).

%%%%%%%%%%%%%%%% BIBLIOGRAPHY IN THE LaTeX file !!!!! %%%%%%%%%%%%%%%%%%%%%%

\end{document}